\documentclass[a4paper, 10pt]{article}

\usepackage{geometry}
\usepackage{authblk}
\usepackage{amsthm}
\usepackage{multirow}
\usepackage{comment}
\usepackage{amssymb}
\usepackage{bm}
\usepackage{ascmac}
\usepackage{algorithm,algpseudocode,float}
\usepackage{color}
\usepackage{breqn}
\usepackage{xspace}
\usepackage{arydshln}
\usepackage[normalem]{ulem}
\usepackage{url}
\usepackage{hyperref}

\newcommand{\TIME}[1]{t_{\fontsize{8pt}{0cm}\selectfont \mathit{\texttt{#1}}}}

\newtheorem{theorem}{Theorem}[section]
\newtheorem{lemma}{Lemma}[section]
\newtheorem{definition}{Definition}[section]
\newtheorem{observation}{Observation}[section]
\newtheorem{corollary}{Corollary}[section]

\title{\textbf{Self-Stabilizing Weakly Byzantine Perpetual Gathering of Mobile Agents}}

\author[1]{Jion Hirose}
\author[2]{Ryota Eguchi}
\author[3]{Yuichi Sudo}

\affil[1]{National Institute of Technology, Ishikawa College}
\affil[2]{Nara Institute of Science and Technology}
\affil[3]{Hosei University}

\date{}

\begin{document}

\maketitle

\begin{abstract}
We study the \emph{Byzantine} gathering problem involving $k$ mobile agents with unique identifiers (IDs), $f$ of which are Byzantine.
These agents start the execution of a common algorithm from (possibly different) nodes in an $n$-node network, potentially starting at different times.
Once started, the agents operate in synchronous rounds.
We focus on \emph{weakly} Byzantine environments, where Byzantine agents can behave arbitrarily but cannot falsify their IDs.
The goal is for all \emph{non-Byzantine} agents to eventually terminate at a single node simultaneously.

In this paper, we first prove two impossibility results: (1) for any number of non-Byzantine agents, no algorithm can solve this problem without global knowledge of the network size or the number of agents, and (2) no self-stabilizing algorithm exists if $k\leq 2f$ even with $n$, $k$, $f$, and the length $\Lambda_g$ of the largest ID among IDs of non-Byzantine agents, where the self-stabilizing algorithm enables agents to gather starting from arbitrary (inconsistent) initial states.
Next, based on these results, we introduce a \emph{perpetual gathering} problem and propose a self-stabilizing algorithm for this problem.
This problem requires that all non-Byzantine agents always be co-located from a certain time onwards.
If the agents know $\Lambda_g$ and upper bounds $N$, $K$, $F$ on $n$, $k$, $f$, the proposed algorithm works in $O(K\cdot F\cdot \Lambda_g\cdot X(N))$ rounds, where $X(n)$ is the time required to visit all nodes in a $n$-nodes network.
Our results indicate that while no algorithm can solve the original self-stabilizing gathering problem for any $k$ and $f$ even with \emph{exact} global knowledge of the network size and the number of agents, the self-stabilizing perpetual gathering problem can always be solved with just upper bounds on this knowledge.
\end{abstract}

\section{Introduction}
\label{sec:Introduction}
\subsection{Background}
\label{subsec:Background}
Mobile agents, or simply agents, are software programs that can autonomously traverse a network, modeled as an undirected graph. 
We focus on the gathering problem which requires multiple agents, initially scattered throughout the network, to eventually terminate at a single node simultaneously.
This problem is crucial for information exchange and enabling more complex behavior, and has been studied extensively \cite{Pelc2019}.

In the literature, \emph{transient} and \emph{Byzantine} faults are considered in the gathering problem.
The transient fault model involves temporary memory corruption, erroneous initialization, or other similar issues in faulty agents.
To overcome transient faults, Ooshita, Datta, and Masuzawa \cite{Ooshita2017} proposed a self-stabilizing algorithm that enables agents to gather from arbitrary (inconsistent) initial states, specifically addressing the gathering problem for two agents ($k=2$), known as the \emph{rendezvous} problem.
They also claim that the algorithm can be extended to a self-stabilizing algorithm for cases where $k\geq 2$, if $k$ is given to agents as global knowledge. 
In the Byzantine fault model, some agents incur Byzantine faults, which are called Byzantine agents, and they can behave arbitrarily against their algorithms; 
thus, this fault model includes other agent faults.
In this paper, we assume that Byzantine agents cannot falsify their identifiers (IDs), which is a common assumption in many existing works.
Byzantine agents under this restriction are sometimes referred to as \emph{weakly} Byzantine agents in the literature, but we simply refer to them as Byzantine agents throughout this paper for simplicity.
To tolerate Byzantine faults, many gathering algorithms \cite{Dieudonne2014, Hirose2022, Hirose2024} have been proposed.
These algorithms require that agents initially know the number $n$ of nodes or a given upper bound $N$ on $n$.
In particular, Dieudonn\'{e}, Pelc, and Peleg \cite{Dieudonne2014} proposed an algorithm for the Byzantine environment if agents know $n$.

Both algorithms in \cite{Dieudonne2014} and \cite{Ooshita2017} share a common characteristic: 
they require global knowledge of the network size or the number of agents to solve the problem. 
The authors in \cite{Ooshita2017} show that without $k$, no self-stabilizing algorithm can solve the gathering problem.
In the Byzantine fault model, it is unclear if the problem is solvable without $n$ or $F$.
Additionally, although these faults have received much attention so far, to the best of our knowledge, no studies consider both fault models \emph{simultaneously} for agent systems.

\label{subsec:OurContribution}
\begin{table}[t]
  \caption{A summary of synchronous Byzantine or self-stabilizing gathering algorithms, assuming that agents have unique IDs. Here, $n$ is the number of nodes, $k$ is the total number of agents, $K$ is a given upper bound of $k$, $f$ is the number of Byzantine agents, $F$ is a given upper bound of $f$, and $\Lambda_g$ is the length of the largest ID among IDs of non-Byzantine agents. Symbols \texttt{REN}, \texttt{GAT}, \texttt{PGAT}, and GK represent rendezvous, gathering, perpetual gathering, and global knowledge, respectively.}
  \label{tab:Contributions}
  \newlength{\myheighta}
  \setlength{\myheighta}{18pt}
  \newlength{\myheightb}
  \setlength{\myheightb}{36pt}
  \centering
  \scalebox{0.85}{
  \begin{tabular}{c|c|c|c|c|c|c}
    \hline
    \parbox[c][\myheighta][c]{0cm}{}Problem&\texttt{REN}&\multicolumn{3}{c|}{\texttt{GAT}}&\multicolumn{2}{c}{\texttt{PGAT}}\\\hline
    \parbox[c][\myheighta][c]{0cm}{}Byzantine&No&No&\multicolumn{2}{c|}{Yes}&Yes&Yes\\\hline
    \parbox[c][\myheighta][c]{0cm}{}Self stabilizing&Yes&Yes&No&Yes&No&Yes\\\hline
    \parbox[c][\myheightb][c]{0cm}{}\multirow{4}{*}{With GK}&trivial&\begin{tabular}{c}Possible\\\cite{Ooshita2017}\end{tabular}&\begin{tabular}{c}Possible\\\cite{Dieudonne2014}\end{tabular}&\begin{tabular}{c}\textbf{Impossible} \\\textbf{for} $k\leq 2f$\\(Result 2)\end{tabular}&trivial&\begin{tabular}{c}\textbf{Possible}\\(Result 3)\end{tabular}\\\cdashline{2-7}[1pt/1pt]
    \parbox[c][\myheightb][c]{0cm}{}&-&$k$&$n$ or $F$&\begin{tabular}{c}even with \\ $n$, $k$, $f$, and $\Lambda_g$\end{tabular}&-&$N$, $K$, $F$, $\Lambda_g$\\\hline
    \parbox[c][\myheightb][c]{0cm}{}Without GK&\begin{tabular}{c}Possible \\\cite{Ooshita2017}\end{tabular}&\begin{tabular}{c}Impossible \\\cite{Ooshita2017}\end{tabular}&\begin{tabular}{c}\textbf{Impossible} \\(Result 1)\end{tabular}&\begin{tabular}{c}Impossible \\\cite{Ooshita2017}\end{tabular}&\begin{tabular}{c}Possible \\\cite{Dieudonne2014}\end{tabular}&- \\\hline
  \end{tabular}
  }
\end{table}

\subsection{Our Contribution}
In this paper, we investigate the following questions: 
(1) is the gathering problem solvable in weakly Byzantine environments \textit{without} global knowledge of the network size or the number of (Byzantine) agents? 
(2) is there a \emph{self-stabilizing} algorithm for the gathering problem? 
(3) if not, is there a relaxed variant of the gathering problem that allows the existence of the self-stabilizing algorithm? 

For the first two questions, we demonstrate the following impossibility results for the gathering problem in weakly Byzantine environments: 
(i) for any number of non-Byzantine agents, no algorithm can solve this problem \textit{without} any global knowledge of the network size or the number of agents, and (ii) when $k\leq 2f$, no self-stabilizing algorithm exists for this problem, even if agents know the number $n$ of nodes, the number $k$ of agents, the number $f$ of Byzantine agents, and the length $\Lambda_g$ of the largest ID among IDs of non-Byzantine agents.
Note that the impossibility proof for (i) is based on the scenario where only one non-Byzantine agent exists in the network, the agent should claim the termination of the algorithm.
This assumption is common in fault-tolerant gathering problems (e.g., \cite{Dieudonne2014, Pelc2018}).

Based on these impossibilities, we propose a relaxed variant of the gathering problem, called \emph{perpetual gathering} problem, in weakly Byzantine environments.
This problem requires that all non-Byzantine agents always be co-located from a certain time $t$ onwards.
Unlike the original gathering problem, this does not require the agents to stay at a single node permanently from some time;
they may move but must do so together along the same edge after time $t$.
The term ``perpetual gathering'' is inspired by the famous ``perpetual exploration'' problem, where an agent repeatedly visits all nodes in the graph (cf. \cite{Das2019}).

For the perpetual gathering problem, if the agents know $\Lambda_g$ and upper bounds $N, K, F$ on $n, k, f$, we propose a self-stabilizing algorithm for any $k$ and $f$ with time complexity of $O(K\cdot F\cdot \Lambda_g\cdot X(N))$, where $X(n)$ represents the time required to visit all nodes in a $n$-nodes network.
For example, $X(n)=O(n)$ \cite{Koucky2002} for a cycle, a clique, and a tree.
For an arbitrary graph, $X(n)=n^5\log n$ \cite{Reingold2008}.
When each node is equipped with a local memory (i.e., a whiteboard), $X(n)=O(m+nD)$, where $m$ is the number of edges and $D$ is the diameter of the graph \cite{Sudo2020}.

Our results indicate that while no algorithm can solve the original self-stabilizing gathering problem for any $k$ and $f$, even with \emph{exact} global knowledge of the network size and the number of agents, the self-stabilizing perpetual gathering problem can always be solved with just upper bounds on this knowledge.
Note that the (non-self-stabilizing) perpetual gathering problem without global knowledge in weakly Byzantine environments can be solved using an algorithm proposed by Dieudonn{\'{e}} et al.~\cite{Dieudonne2014}.
If the termination detection mechanism is removed from this algorithm, global knowledge becomes unnecessary.
In this case, agents continue to move, but after a certain time, all non-Byzantine agents meet at the same node.
This can be interpreted as solving the perpetual gathering problem without global knowledge.
We summarize our contribution and existing results in Table \ref{tab:Contributions}.

\subsection{Related Work}
\label{subsec:RelatedWorks}
Many researchers considered rendezvous and gathering problems in various scenarios, such as synchrony, anonymity, network topology, randomness, and so on, and clarified the solvability.
They also proposed various solutions to reduce various costs (e.g., time complexity, the number of moves, memory capacity, etc.).
In this paper, we focus on deterministic solutions for rendezvous and gathering problems in the scenario where agents have unique IDs and move synchronously; 
Pelc \cite{Pelc2019} extensively surveyed these solutions, including those for anonymous agents and asynchronous environments.
A comprehensive survey of the randomized solutions can also be found in the book by Alpern and Gal \cite{Alpern2002}.

In such a scenario, the majority of the results are for fault-free settings.
These results exploited agent IDs to break the symmetry. 
Several studies \cite{Dessmark2006,Kowalski2008,Ta-Shma2014} proposed algorithms to solve rendezvous problems for arbitrary graphs if agents do not know the number of nodes $n$ and start with a delay of at most $\tau$ rounds.
Dessmark et al.~\cite{Dessmark2006} presented the first deterministic algorithm, with polynomial time in $n$, $\tau$, and $\lambda$, where $\lambda$ is the length of the smallest ID among agent IDs.
Kowalski and Malinowski \cite{Kowalski2008} proposed an algorithm whose time complexity is independent of $\tau$ and is polynomial in $n$ and $\lambda$.
Ta{-}Shma and Zwick \cite{Ta-Shma2014} provided an algorithm with lower time complexity that is polynomial in $n$ and $\lambda$.
Miller and Pelc \cite{Miller2016} investigated tradeoffs between the number of rounds and the total number of edge traversals until the meeting.

Elouasbi and Pelc \cite{Elouasbi2017} introduced a new problem (called gathering with detection), which requires agents to declare the termination simultaneously when the meeting is perceived only through communication.
They proposed an algorithm for this problem on arbitrary graphs in polynomial time in $n$ and $\lambda$ if agents use beeps as their communication method and do not know $n$.
Bouchard et al.~\cite{Bouchard2023} considered a weaker model where each agent only detects the number of agents located at the same node in a round.
They provided two types of algorithms to achieve the gathering with detection on arbitrary graphs in this model:
one in polynomial time in $N$ and $\lambda$ when agents know the upper bound $N$ of $n$, and the other in exponential time in $n$ when agents do not know $n$.
Molla et al.~\cite{Molla2023} proposed a faster algorithm to solve the gathering with detection on arbitrary graphs.
Their algorithm is faster than the existing algorithms for even just gathering if agents do not know $n$ and start simultaneously, with polynomial time in $n$.

Recently, the gathering under various types of agent faults has been studied, particularly the gathering problem in the presence of Byzantine agents.
This problem assumes that the network topology is arbitrary, there are $k$ agents, and $f$ of them are Byzantine.
Studies addressing these problems considered two types of Byzantine agents: strongly and weakly Byzantine.
Strongly Byzantine agents can ignore the algorithm and behave arbitrarily, while weakly Byzantine agents can do the same except falsify their own IDs. 
Dieudonn{\'{e}} et al.~\cite{Dieudonne2014} designed the first gathering algorithms to tolerate Byzantine agents.
They proposed algorithms to solve the gathering for weakly Byzantine agents in polynomial time in $n$ and $\Lambda_g$, where $\Lambda_g$ is the length of the largest ID among IDs of non-faulty agents, both if agents know $n$ and $k\geq f+1$ holds, and if agents know the upper bound $F$ of $f$ and $k\geq 2F+2$ holds. 
These algorithms match the lower bounds on the number of non-faulty agents required.
They also provided algorithms to solve the gathering for strongly Byzantine agents in exponential time in $n$ and $\Lambda_g$, both if agents know $n$ and $F$ and $k\geq 3F+1$ holds, and if agents know $F$ and $k\geq 5F+2$ holds.
However, the lower bounds in these cases are $F+1$ and $F+2$, respectively, which means that the required numbers of non-faulty agents required by these algorithms do not match these lower bounds.
Bouchard et al.~\cite{Bouchard2016} improved the upper bounds to $k\geq 2F+1$ if agents know $n$ and $F$, and to $k\geq 2F+2$ if agents know $F$.
Bouchard et al.~\cite{Bouchard2022} provided an algorithm that achieves the gathering in the presence of strongly Byzantine agents in polynomial time in $N$ and $\lambda$ if agents have global knowledge of $O(\log\log\log N)$ and $k\geq 5f^2+7f+2$ holds.
Hirose et al.~\cite{Hirose2022} shown a gathering algorithm that tolerates weakly Byzantine agents if agents know $N$ and $k\geq 4f^2+9f+4$ holds, with time complexity smaller than \cite{Dieudonne2014} and polynomial in $N$, $f$, and $\Lambda_a$, where $\Lambda_a$ is the length of the largest ID among agent IDs.
Additionally, Hirose et al.~\cite{Hirose2024} proposed a similar algorithm if agents know $N$ and $k\geq 9f+8$ holds, with time complexity greater than \cite{Hirose2022} but smaller than \cite{Dieudonne2014}.
Tsuchida et al.~\cite{Tsuchida2018} reduced the time complexity of the gathering algorithm to tolerate weakly Byzantine agents by assuming authenticated whiteboards, where each node has dedicated memory for each agent to leave information if agents know $F$ and $k\geq F$ holds.
Their algorithm works in polynomial time in $F$ and $m$, where $m$ is the number of edges.
Tsuchida et al.~\cite{Tsuchida2020} showed an algorithm to solve the gathering problem in asynchronous environments with weakly Byzantine agents by using authenticated whiteboards.
Miller and Saha \cite{Miller2020} considered the case where agents can obtain information in the subgraph induced by nodes within a radius of the network from their current node.
Their proposed algorithm solves the gathering problem in the presence of strongly Byzantine agents in polynomial time in $k$ and $n$ if $k\geq 2f+1$ holds.

Several studies looked at other types of agent faults.
Chalopin et al.~\cite{Chalopin2016} studied delay faults, where faulty agents remain at the current node for a finite number of rounds regardless of their plans to move.
Pelc \cite{Pelc2018} considered the gathering problem with crash faults, where some agents suddenly become immobile at the node in some round, for systems where each agent moves at constant speed but their speeds differ.
Ooshita et al.~\cite{Ooshita2017} proposed a self-stabilizing algorithm to achieve the rendezvous when agents start with arbitrary states of their memory.
This self-stabilizing algorithm guarantees that agents eventually meet at a single node even if their memory experiences any kind of corruption, i.e., it counteracts transient faults in agent memory.

\section{Preliminaries}
\subsection{Model}
The system is represented as a simple, connected, and undirected graph $G=(V,E)$ with $n=|V|$ nodes.
Nodes have no ID.
Each node $v\in V$ assigns unique port numbers in $\{1,\dots,d(v)\}$ to its edges, where $d(v)$ is the degree of $v$.
Port numbers are local, meaning edges $(v,u)$ and $(u,v)$ for nodes $v$ and $u$ may have different numbers.

There are $k$ agents in the system, denoted by set $A$.
Each agent $a\in A$ has a unique ID $a.id$ but knows only its own ID.
We denote ``an agent with an ID in an ID set $S_{id}$'' as ``an agent in $S_{id}$'' for brevity.
Agents have unbounded memory but cannot leave any information on nodes or edges.
Among the $k$ agents, exactly $f$ are Byzantine, which means they can act arbitrarily but cannot falsify their IDs; 
thus, when a Byzantine agent $b$ communicates with another agent $a$, $a$ can correctly identify $b$'s ID. 
These agents are commonly referred to as \emph{weakly Byzantine agents} and have been studied in the literature, including \cite{Dieudonne2014,Hirose2022,Hirose2024}.
We call all the non-Byzantine agents \emph{good} agents.
We denote the length of the largest ID in IDs of good agents as $\Lambda_g$.

An algorithm $\mathcal{A}$ for an agent $a$ is defined as $\mathcal{A}(a.id, n',k',f',\Lambda_g')=(S,\delta,s_{ini},S_{ter})$, where $S$ is a set of states, $\delta$ is a function for transitioning its state, $s_{ini}\in S$ is the special state, and $S_{ter}$ is a set of terminal states.
Algorithm $\mathcal{A}$ takes five arguments, which are assigned integer values.
The last four arguments correspond to $n, k, f$, and $\Lambda_g$, and each value is assumed to be assigned if required by the algorithm. 
We refer to the arguments $n', k', f'$ and $\Lambda_g'$ as \emph{global knowledge}, and we assume that each of these arguments takes the same value for all agents in the system.
In the above tuple, only the function $\delta$ depends on $a.id$, which means that the states of all agents are identical.
States are represented by a tuple of variables of $a$; 
if its variables contain associative arrays, we refer to these simply as arrays.
State $s_{ini}$ represents one where all variables of $a$ are initialized.
A subset $S_{ter}$ will be defined later. 

A \emph{configuration} in the system is defined as a combination of the locations, the states, and the IDs of all agents in $A$.
We define $C_{all}$ as the set of all possible configurations in which at least one good agent is not in a dormant state.
An agent in a dormant state, or simply a dormant agent, transitions into $s_{ini}$ when (a) the adversary wakes up the dormant agent, or (b) a non-dormant agent visits the node with the dormant agent.
Non-dormant agents execute a common algorithm in synchronous and discrete time steps, called \emph{rounds}.
We let $C_{ini}$ be a set of configurations in $C_{all}$ where each agent is either dormant or in $s_{ini}$.

In each round $r$, each non-dormant agent $a_i$ at $v$ performs the following operations sequentially:
\begin{enumerate} 
  \item[(1)] It acquires $d(v)$, the port number $p_{in}$ through which it arrived at $v$ in $r-1$ (if it stayed in $r-1$, $p_{in}=\bot$), and states of agents (including $a_i$) at $v$. 
  We define $\mathit{MG}_i$ as the set of $a_i$ and agents that $a_i$ observes at $v$ in $r$.
  \item[(2)] It updates its next state $s'$ using $\delta$ with $d(v)$, $p_{in}$, the states of all agents in $\mathit{MG}_i$, and the arguments of $\mathcal{A}$, deciding whether to stay or leave, and determining the outgoing port $p_{out}$.
  If it decides to stay, $p_{out} = \bot$. 
  \item[(3)] It either stays at $v$ or moves to the destination node via $p_{out}$ by $r+1$.
\end{enumerate}
Note that in Operation (1), every agent in $\mathit{MG}_i$ has to share its states with others before any agent in $\mathit{MG}_i$ computes $\delta$, but Byzantine agents in $\mathit{MG}_i$ can falsify the values of their variables to others.
We assume that all agents in $\mathit{MG}_i$ observe the same state from any Byzantine agent in $\mathit{MG}_i$, i.e., Byzantine agents share the same (possibly falsified) states with all co-located agents.
This assumption is known as \emph{shouting} and has been adopted in previous studies on Byzantine gathering (e.g., \cite{Dieudonne2014,Bouchard2016,Bouchard2022,Hirose2022,Hirose2024}).
Without this assumption, a Byzantine agent could send different states to different co-located agents.
This makes it significantly harder for co-located good agents to detect inconsistencies or coordinate their actions. 
Note also that in Operation (3), if two agents pass the same edge in opposite directions simultaneously, they do not notice this fact.
If an agent enters a state in $S_{ter}$ at node $v$ in round $r$, $\delta$ outputs $p_{out} = \bot$ and $s' \in S_{ter}$ in every round from $r$ onwards.
Since the terminal states depend on the global knowledge, we sometimes denote $S_{ter}$ by $S_{ter}(n', k', f', \Lambda'_g)$ if needed.
An execution starting from $c_0\in C_{all}$ is a sequence $c_0, c_1, c_2, \dots$, such that each configuration at round $r$ is determined (deterministically) by applying the procedures described above to all the good agents and determining the behavior of the Byzantine agents by the adversary. 

\subsection{Problems}
\label{subsec:Problems}
\paragraph*{Gathering problem.} 
This problem requires all good agents to eventually enter a state in $S_{ter}$ at a single node simultaneously.

\begin{definition}
\label{def:GatheringAlgorithm}
For a graph $G = (V, E)$ and $C^* \subseteq C_{all}$, an algorithm $\mathcal{A}_{\mathit{g}}$ solves the gathering problem in $G$ from $C^*$ if all good agents, executing $\mathcal{A}_{\mathit{g}}$ in $G$, eventually enter terminal states at the same node simultaneously, starting from any initial configuration of $C^*$.
We say an algorithm $\mathcal{A}_g$ solves a gathering problem (\texttt{GAT}) if for any $G = (V, E)$ and $C_{ini}$, it solves the problem in $G$ from any of $C_{ini}$. 
Similarly, we say an algorithm $\mathcal{A}_{ssg}$ solves a self-stabilizing gathering problem (\texttt{SS-GAT}) if for any $G = (V, E)$ and $C_{all}$, it solves the problem in $G$ from any of $C_{all}$.
\end{definition}

\paragraph*{Perpetual gathering problem.} 
This problem requires all good agents to always be co-located from a certain round $r$ onwards.
Unlike the gathering problem, this problem does not require that these agents stay at a single node permanently;
they may move but must do so together along the same edge after round $r$.

\begin{definition}
\label{def:SSPGAlgorithm}
For a graph $G = (V, E)$ and $C^* \subseteq C_{all}$, an algorithm $\mathcal{A}_{\mathit{pg}}$ solves the perpetual gathering problem (\texttt{PGAT}) in $G$ from $C^*$ if all good agents, executing $\mathcal{A}_{\mathit{pg}}$ in $G$, always stay at a single node from a certain round onwards, starting from any initial configuration of $C^*$.
We say an algorithm $\mathcal{A}_{sspg}$ solves a self-stabilizing gathering problem (\texttt{SS-PGAT}) if for any $G = (V, E)$ and $C_{all}$, the algorithm solves the perpetual gathering problem in $G$ from any of $C_{all}$.
\end{definition}

\paragraph*{Global knowledge and rounds.}

\begin{definition}
For a problem \texttt{P} $\in \{ \texttt{GAT, PGAT, SS-GAT, SS-PGAT} \}$, we say an algorithm $\mathcal{A}$ solves \texttt{P} with global knowledge $n, k, f$ and $\Lambda_g$ if $\mathcal{A}(id, n', k', f', \Lambda'_g)$ solves \texttt{P} with the assumption $(n', k', f', \Lambda'_g) = (n, k, f, \Lambda_g)$. 
Similarly, we say $\mathcal{A}$ solves \texttt{P} with global knowledge $N, K, F$ and $\Lambda_g$ if $\mathcal{A}(id, n', k', f', \Lambda'_g)$ solves \texttt{P} with the assumption $(n', k', f', \Lambda'_g) = (N, K, F, \Lambda_g)$ for all $N, K, F$ such that $n \le N, k \le K, f \le F$.
Additionally, we say an algorithm $\mathcal{A}$ solves \texttt{P} without the global knowledge, if $\mathcal{A}(id, n', k', f', \Lambda'_g)$ solves \texttt{P} for any values of $(n', k', f', \Lambda'_g)$.
\end{definition}

We evaluate the time complexity of each algorithm by the number of rounds required for any execution to reach the target configuration from the first configuration in which at least one good agent is not dormant.

\subsection{Existing Algorithm}
In the proposed algorithm, we incorporate a rendezvous algorithm as a subroutine to ensure that an agent meets other agents.
This algorithm, due to Ta{-}Shma and Zwick~\cite{Ta-Shma2014}, enables two agents to meet in any undirected and connected graph of $n$ nodes, starting from different nodes.
We refer to this algorithm as $\mathsf{REN}(\mathit{label})$, where $\mathit{label}$ is a positive integer given as input.
When an agent $a_i$ starts $\mathsf{REN}(\mathit{label}_i)$ for some positive integer $\mathit{label}_i$, it alternates between exploring and waiting periods based on $\mathit{label}_i$.
The agent executes an exploring algorithm using a universal exploration sequence, due to Reingold\cite{Reingold2008}, during the exploring period and waits at the node for the duration required by the exploring algorithm during the waiting period.
The scheduling of these periods ensures that two agents with different positive integers meet, typically when one agent waits while the other explores.
When two agents with distinct positive integers $\mathit{label}_i$ and $\mathit{label}_j$ execute this procedure, this time complexity of $\mathsf{REN}(\mathit{label})$, denoted by $\TIME{REN}(\mathit{label})$, is $\Tilde{O}(n^5\cdot\lambda)$, where $\lambda$ is the length of $\min(\mathit{label}_i,\mathit{label}_j)$, and $\Tilde{O}$ hides a poly-logarithimic factor of $n$.
We have the following theorem for this algorithm.

\begin{theorem} (Theorems 2.3 and 3.2 of Ta{-}Shma and Zwick \cite{Ta-Shma2014})
\label{the:RendezvousAlgorithm}
Let $a_i$ and $a_j$ be two agents, and $\mathit{label}_i$ and $\mathit{label}_j$ be positive integers such that $\mathit{label}_i\neq \mathit{label}_j$ holds.
If $a_i$ and $a_j$ start $\mathsf{REN}(\mathit{label}_i)$ and $\mathsf{REN}(\mathit{label}_j)$ in rounds $r_i$ and $r_j$, respectively, they meet at a single node by round $\max(r_i,r_j)+\TIME{REN}(\min(\mathit{label}_i,\mathit{label}_j))$.
\end{theorem}

\noindent We denote the $t$-th round of $\mathsf{REN}(id)$ by $\mathsf{REN}(id,t)$ for integer $t\geq 0$.

\section{Impossibility results}
\label{sec:ImpossibilityResults}
In this section, we present the following impossibility results for the solvabilities of the \texttt{GAT} and the \texttt{SS-GAT}: there is no gathering algorithm without global knowledge, and there is no self-stabilizing gathering algorithm with $n$, $k$, $f$, and $\Lambda_g$.

\subsection{No gathering algorithm without global knowledge}
In this section, we show that no algorithm solves the \texttt{GAT} without global knowledge.
To establish this result, we prove that without global knowledge, an agent cannot visit all nodes of a ring and declare the termination. 
While this proof is well-known, we provide a rigorous proof for the result since it is critical for demonstrating the impossibility of solving the \texttt{GAT}.

\begin{lemma}
\label{lem:NoExplorationAlgorithmsWithoutGlobalKnowledge}
There is no algorithm without global knowledge for a single agent to visit all nodes of a ring and declare the termination after visiting all nodes.
\end{lemma}
\begin{proof}
We prove this lemma by contradiction.
Suppose such an algorithm $\mathcal{A}_e$ exists.
Let $R_3$ be a ring composed of 3 nodes, with each edge labeled with port numbers in the clockwise direction.
When an agent $a_i$ starts $\mathcal{A}_{e}$ from a node of $R_3$, it visits all nodes of $R_3$ and declares the termination within $T$ rounds for some integer $T$.
Consider a ring $R_\ell$ composed of $\ell > T$ nodes, where each edge is labeled with port numbers in the clockwise direction, similar to $R_3$.
If $a_i$ starts $\mathcal{A}_e$ from a node of $R_\ell$, $a_i$ declares the termination after $T$ rounds, despite not having visited all nodes in $R_\ell$, because $a_i$ cannot distinguish $R_3$ and $R_\ell$.
This contradiction proves the lemma.
\end{proof}

Next, we use the above lemma to prove that no algorithm can solve the \texttt{GAT} without global knowledge. 

\begin{theorem}
\label{the:ImpossibilityGathering}
For any number $g\geq 1$ of good agents, no algorithm solves the \texttt{GAT} without any global knowledge.
\end{theorem}
\begin{proof}
To derive a contradiction, we assume that such an algorithm exists $\mathcal{A}_g(id,n',k',f',\Lambda_g')$.
Since the algorithm works without any global knowledge, we denote it simply as $\mathcal{A}_g(id)$ in the following discussion.
We analyze two cases on the number $g$ of good agents.

\textbf{Case 1 ($g=1$):}
Suppose that there is only one good agent in the network.
Then, by the definition of the \texttt{GAT}, this agent eventually transitions into a terminal state at some node, regardless of whether it meets Byzantine agents.
Assume that there exists an ID $id^*$ such that, for any graph $G = (V, E)$ and any initial position $v \in V$, when the good agent with ID $id^*$ executes $\mathcal{A}_g(id^*)$, it always visits all nodes in $V$ before terminating.
Note that this must hold even if the Byzantine agent does not interact with the good agent.
This implies that there could exist an exploring algorithm with a single agent by simulating $\mathcal{A}_g$ with $id^*$ for $G$ and $v$. 
This contradicts Lemma \ref{lem:NoExplorationAlgorithmsWithoutGlobalKnowledge}.

\textbf{Case 2 ($g\geq 2$):}
Suppose that there are at least two good agents in the network.
Assume that no ID is guaranteed to visit all nodes in the network before terminating. 
That is, for every ID $id$, there exists a graph $G = (V, E)$ such that when a good agent starts $\mathcal{A}_g$, the good agent enters a terminal state \emph{before} visiting all nodes in $V$.
We first consider the case $g=2$.
Let $a_i$ and $a_j$ be good agents with different IDs.
By the assumption above, there exists a graph $G_i = (V_i, E_i)$ and a starting node $v_i \in V_i$ such that $a_i$, starting from $v_i$, terminates without visiting some node $u_i \in V_i$.
Similarly, there exists a graph $G_j = (V_j, E_j)$ and node $v_j \in V_j$ such that $a_j$, starting from $v_j$, terminates without visiting some node $u_j \in V_j$.
Consider a graph $G_{ij}=(V_{ij},E_{ij})$ where $V_{ij}=V_i\cup V_j$ and $E_{ij}=E_i\cup E_j\cup \{(u_i,u_j)\}$.
In $G_{ij}$, when $a_i$ and $a_j$ start $\mathcal{A}_g(a_i.id)$ and $\mathcal{A}_g(a_j.id)$ from $v_i$ and $v_j$ respectively, $a_i$ never visits $u_i$ and $a_j$ never visits $u_j$.
Since $u_i$ and $u_j$ are the only connection between $V_i$ and $V_j$, two agents terminate at different nodes without ever meeting, which is a contradiction.
For the case $g \geq 3$, the same contradiction can be derived by applying the same construction used in the case $g = 2$ to any pair of agents.

Hence, no such algorithm $\mathcal{A}_g$ can exist for $g$.
\end{proof}

Note that some algorithms, e.g., considered in \cite{Dieudonne2014, Pelc2018}, have the basic property that they terminate even if there is only a good agent.

\subsection{No self-stabilizing gathering algorithm with \texorpdfstring{$n$, $k$, $f$ and $\Lambda_g$}{n, k, f and Lambda\_g}}
In this section, we prove that no algorithm can solve \texttt{SS-GAT} with $n$, $k$, $f$, and $\Lambda_g$.

\begin{theorem}
\label{the:ImpossibilitySSGathering}
Let $n$ be the number of nodes, $k$ be the number of agents, $f$ be the number of Byzantine agents, and $\Lambda_g$ be the length of the largest ID among IDs of good agents.
When $k\leq 2f$, no algorithm solves the \texttt{SS-GAT} even with the global knowledge $n$, $k$, $f$, and $\Lambda_g$.
\end{theorem}
\begin{proof}
To derive the contradiction, we assume that there exists an algorithm $\mathcal{A}_{ssg}(id,n',k',f',\Lambda_g')$ with $(n', k', f', \Lambda_g') = (n, k, f, \Lambda_g)$ that solves the problem.
We suppose that $n$ and $f$ are even and $k\leq 2f$ holds.
Let $I_\alpha$ and $I_\beta$ be disjoint sets of agent IDs such that $|I_\alpha|=|I_\beta|\leq f$ holds and each length of IDs in $I_\alpha\cup I_\beta$ is $\Lambda_g$. 
Such $I_{\alpha}$ and $I_{\beta}$ should exist when $\Lambda_g \ge \lceil \log (2f) \rceil + 1$.

Consider a graph $G_1=(V_1,E_1)$ with $n$ nodes where the degree of each node is at most $n/2-1$.
Suppose that all agents in $I_\alpha$ are good and all agents in $I_\beta$ are Byzantine in $G_1$.
When good agents execute $\mathcal{A}_{ssg}(id,n',k',f',\Lambda_g')$ starting from $c \in C_{all}$, $|I_\alpha|$ good agents in $I_\alpha$ eventually meet and enter terminal states in $S_{ter}(n', k', f', \Lambda_g')$ at a node $v_1$. 
Note that these executions include the situation that each good agent does not meet any Byzantine agent; 
therefore, we can assume that there are only $|I_\alpha|\leq f$ good agents in $v_1$ after these agents enter terminal states.
Let such good agents in $I_\alpha$ be $a^1_0, a^1_1, \dots, a^1_{|I_\alpha|-1}$.

Similarly, consider another graph $G_2=(V_2, E_2)$ with $n$ nodes where the degree of each node in $G_2$ is at most $n/2-1$.
Suppose that all agents in $I_\alpha$ are Byzantine, and all agents in $I_\beta$ are good in $G_2$.
When good agents execute $\mathcal{A}_{ssg}(id,n',k',f',\Lambda_g')$ starting from $c' \in C_{all}$, $|I_\beta|$ good agents in $I_\beta$ eventually meet and enter terminal states in $S_{ter}(n',k',f',\Lambda_g')$ at a node $v_2$.
As with $G_1$, we assume that there are only $|I_\beta|$ good agents in $v_2$ after these agents enter terminal states.
Let such good agents in $I_\beta$ be $a^2_0, a^2_1, \dots, a^2_{|I_\beta|-1}$.

Let $S_1=(V^S_1,E^S_1)$ (resp. $S_2=(V^S_2,E^S_2)$) be a star whose internal node, denoted by $v'_1$ (resp. $v'_2$), has the same degree as $d(v_1)$ (resp. $d(v_2)$), and $L_1=(V^L_1,E^L_1)$ (resp. $L_2=(V^L_2,E^L_2)$) be a line composed of $n/2-d(v_1)-1$ nodes (resp. $n/2-d(v_2)-1$ nodes).
We consider a tree $T_3=(V^S_1\cup V^L_1\cup V^S_2\cup V^L_2, E^S_1\cup E^L_1\cup E^S_2\cup E^L_2\cup \{(u^S_1, u^S_2),(u^S_1, u^L_1),(u^S_2, u^L_2)\})$, where $u^S_1$ is a leaf in $S_1$, $u^S_2$ is a leaf in $S_2$, $u^L_1$ is a node in $L_1$, and $u^L_2$ is a node in $L_2$.
We also consider a configuration $c^*$ where $v'_1$ hosts $|I_\alpha|/2$ good agents and $|I_\alpha|/2$ Byzantine agents and $v'_2$ hosts $|I_\beta|/2$ good agents and $|I_\beta|/2$ Byzantine agents.
Specifically, at $v'_1$, every agent $a^{1'}_i$ matches $a^1_i$ in both ID ($a^{1'}_i.id = a^1_i.id$) and state ($s^{1'}_i = s^1_i$). 
Similarly, at $v'_2$, each agent $a^{2'}_i$ has an ID and state identical to its counterpart $a^2_i$ ($a^{2'}_i.id = a^2_i.id$ and $s^{2'}_i = s^2_i$).
Since $T_3$, $G_1$, and $G_2$ share the same values $n$, $k$, $f$, and $\Lambda_g$, we can observe that $S_{ter}(n', k', f', \Lambda_g')$ of each agent is identical.
Therefore, when good agents execute $\mathcal{A}_{ssg}(id,n',k',f',\Lambda_g')$ starting from $c^*$, they thereafter remain stationary because their $p_{out}$ and $s'$ are $\bot$ and a state in $S_{ter}(n', k', f', \Lambda_g')$ after the start, respectively.
This implies that the good agents cannot meet at a single node, which is a contradiction.
\end{proof}

\section{Self-stabilizing perpetual gathering algorithm with \texorpdfstring{$N$, $K$, $F$ and $\Lambda_g$}{N, K, F and Lambda\_g}}
\label{sec:SS-PG}
In this section, we provide an overview of our proposed algorithm for solving the \texttt{SS-PGAT} with $N$, $K$, $F$, and $\Lambda_g$.
We then detail the behaviors of the proposed algorithm.
Throughout the explanation and proof for the proposed algorithm, we assume that all good agents start an algorithm from $c_0\in C_{all}$ simultaneously.
We discuss how to remove this assumption in Section \ref{sec:Discussion}.

\subsection{Overview}

We present the underlying idea of the proposed algorithm.
We first show two solutions for the \texttt{PGAT}: one in non-Byzantine environments, where no Byzantine agents exist, and another in Byzantine environments.
The latter solution is based on an idea by Dieudonn\'{e} et al.~\cite{Dieudonne2014}.
We then extend the latter solution to solve the \texttt{SS-PGAT}.
We assume that agents start from any of $C_{ini}$ when explaining the solutions for the \texttt{PGAT}.

To solve the \texttt{PGAT} in non-Byzantine environments, agents behave as follows: 
if an agent $a_i$ is alone at the current node, it starts $\mathsf{REN}(\mathit{seed})$, using $a_i.id$ as the input $\mathit{seed}$ for the rendezvous algorithm.
Whenever $a_i$ meets other agents, $a_i$ stops the execution of the current rendezvous algorithm and starts it using the smallest ID among IDs in $\mathit{MG}_i$ as $\mathit{seed}$.
According to Theorem \ref{the:RendezvousAlgorithm}, the number of agents traveling together grows, eventually ensuring that all good agents stay at the same node.
However, this approach fails in a Byzantine environment.
Consider a Byzantine agent $b$ with an ID smaller than the smallest ID among the good agent IDs.
Agent $b$ repeatedly meets some good agents and leaves them. 
As a result, those good agents restart the rendezvous algorithm each time this action of $b$ occurs. 
Those good agents cannot meet the remaining good agents because this repeated restarting prevents uninterrupted execution for a sufficiently long period. 
Moreover, it is also ineffective for agents to record the IDs of previously left agents and then exclude those recorded IDs when selecting $\mathit{seed}$.
Agent $b$ could leave some good agents and then travel with the remaining good agents.
This causes the good agents to divide into two groups.

This problem can be resolved using the idea of Dieudonn\'{e} et al.~\cite{Dieudonne2014}.
This idea extends the approach in non-Byzantine environments by modifying the determination of $\mathit{seed}$, as follows, to mitigate the influence of Byzantine agents. 
In this idea, a \emph{group} is defined as a set of agents at a single node executing the rendezvous algorithm using the same $\mathit{seed}$.
To ignore Byzantine agents, when the group members change, the agents in the group compute the largest subset $S_b$ of agents at the current node such that every agent outside of $S_b$ has observed each member of $S_b$ leaving the group at some point in the past.
Then, if the group meets agents neither in the group nor in $S_b$, or agents not in $S_b$ leave the group, the group stops the execution of the current rendezvous algorithm and starts it using the smallest ID among IDs in $\mathit{MG}_i$ and not in $S_b$ as $\mathit{seed}$.

In this paper, for future extensions, we represent this idea using a labeled graph $G_{CG}=(V_{CG},E_{CG})$, called \emph{confidence graph}, where for an agent $a_i$, $|V_{CG}|=|\mathit{MG}_i|$, each node label corresponds one-to-one with the ID of an agent in $\mathit{MG}_i$, and each edge $(u,w)\in E_{CG}$ indicates that the two agents with the labels of nodes $u$ and $w$ trust each other.
The detailed behaviors using this graph in each round are as follows:
Agent $a_i$ first checks whether each agent $a_j$ at the current node is trustworthy.
If $a_i$ detects fraud by $a_j$, $a_i$ does not trust $a_j$; otherwise, $a_i$ trusts $a_j$.
Agent $a_i$ then simulates how $a_j$ would evaluate trust relationships.
Once $a_i$ completes the above simulation for all agents in $\mathit{MG}_i$, it derives the trust relationships among those agents in $\mathit{MG}_i$.
Next, $a_i$ maps the trust relationships onto a confidence graph and calculates the IDs reachable from the node with $a_i.id$ in the graph to determine the group members.
If the group members change from the previous round, $a_i$ updates $\mathit{seed}$; otherwise, it continues using the same $\mathit{seed}$.
This behavior ensures that good agents at the same node make the same confidence graph, as they base it on the states of agents at the same node at the beginning of the current round.
Thus, good agents at the same node always select the same ID as $\mathit{seed}$.
Consequently, after meeting, good agents travel together, and eventually, all good agents stay at the same node from a certain round.

While this approach solves the \texttt{PGAT} without global knowledge, it cannot be directly applied to the \texttt{SS-PGAT}, as agents start from a configuration of $C_{all}$, that is, each agent may start with a different state.
In cases where some good agents already have a memory of traveling together with other good agents in the initial configuration, they may never trust those agents.
As a result, the good agents remain divided into several groups.

To address this issue, the proposed algorithm restricts the length $\tau$ of the period that an agent suspects the other agents.
In the proposed algorithm, each agent derives $\tau$ from global knowledge.
All good agents stop suspecting the other good agents after the first $\tau$ rounds.
By appropriately setting $\tau$, we will prove that a good agent meets another good agent within $\tau$ rounds, i.e., before a Byzantine agent suspected by the good agent can regain the trust of the good agent.
Thus, the proposed algorithm ensures that all good agents stay at the same node within $2\tau$ rounds.

\subsection{Details}
\begin{algorithm}[t]
  \caption{$\mathcal{A}_{sspg}(id, n', k', f', \Lambda_g')$}
  \label{alg:SSPGAlgorithm}
  \begin{algorithmic}[1]
    \State $a_i.\mathit{numRound}\gets (a_i.\mathit{numRound}+1)\bmod \TIME{REN}(2^{\Lambda_g+1})$
    \State Execute $\mathsf{UpdateTrustRelationship}()$
    \State Execute $\mathsf{SelectSeed}()$
    \State Execute $\mathsf{REN}(a_i.\mathit{seed},a_i.\mathit{numRound})$
  \end{algorithmic} 
\end{algorithm}

\begin{table}[t]
  \centering
  \caption{Variables of agent $a_i$.}
  \label{tab:variable}
  \begin{tabular}{c|p{0.8\hsize}}\hline
    Variable & \multicolumn{1}{c}{Explanation}\\ \hline
    $\mathit{numRound}$ & The number of rounds since the beginning of the current $\mathsf{REN}$. \\ \hline
    $\mathit{seed}$ & The ID used for the current $\mathsf{REN}$. \\ \hline
    $T_{wit}[id_1][id_2]$ & Each element keeps an estimated round since an agent with $id_1$ witnessed the cheating of an agent with $id_2$, and it takes a value from $1$ to $\tau$. \\ \hline
    $R$ & A set of IDs of group members at the end of the previous round. \\ \hline
  \end{tabular}
\end{table}

\begin{table}[t]
  \centering
  \caption{Functions used in $\mathcal{A}_{sspg}(id, n', k', f', \Lambda_g')$}
  \label{tab:functions}
  \begin{tabular}{c|p{0.7\hsize}}\hline
    Function & \multicolumn{1}{c}{Explanation}\\ \hline
    $\mathrm{OVERWRITE}(S_1,S_2)$ & Overwrites an array $S_1$ (resp. a two-dimensional array $S_1$) with an array $S_2$ (resp. a two-dimensional array $S_2$). \\ \hline
    $\mathrm{REMOVE}(S,i)$ & Removes an index $i$ from indices of $S$ if $S$ is an array, and removes an index $i$ from indices of the first-dimension of $S$ if $S$ is a two-dimensional array. \\ \hline
  \end{tabular}
\end{table}

Algorithm \ref{alg:SSPGAlgorithm} outlines the behavior for each round of Algorithm $\mathcal{A}_{sspg}(id, n', k', f', \Lambda_g')$ and uses Algorithms \ref{alg:UpdateTrustRelationship} and \ref{alg:SelectSeed} as sub-routines.
Tables \ref{tab:variable} and \ref{tab:functions} summarize the variables used by an agent $a_i$ in $\mathcal{A}_{sspg}(id, n', k', f', \Lambda_g')$ and functions for arrays in $\mathcal{A}_{sspg}(id, n', k', f', \Lambda_g')$, respectively.
In $\mathcal{A}_{sspg}(id, n', k', f', \Lambda_g')$, $i\in S_1$ for an index $i$ and an array $S_1$, and $j\in S_2$ for an index $j$ and a two-dimensional array $S_2$, indicates that $i$ is included in the indices of $S_1$ and $j$ is included in the indices of first dimension of $S_2$, respectively.
Additionally, when an agent assigns a value $c$ to $S[i]$ for an array $S$ and a non-existent index $i$, we simply describe $S[i]\gets c$.
Each agent calculates $\tau=(6KF+1)\cdot \TIME{REN}(2^{\Lambda_g+1})$ based on global knowledge $N$, $K$, $F$, and $\Lambda_g$.

We focus on the process of each round by an agent $a_i$.
First, $a_i$ increments the value of $a_i.\mathit{numRound}$ by one.
If the value assigned to $a_i.\mathit{numRound}$ exceeds $\TIME{REN}(2^{\Lambda_g+1})$, $a_i$ stores the value modulo $\TIME{REN}(2^{\Lambda_g+1})$ in $a_i.\mathit{numRound}$.
Next, using Algorithm \ref{alg:UpdateTrustRelationship}, $a_i$ calculates the trust relationship of all agents in $\mathit{MG}_i$.
Then, using Algorithm \ref{alg:SelectSeed}, $a_i$ selects an ID $id_{seed}$ as the input of the rendezvous algorithm and reaches a consensus on $\mathit{numRound}$ among agents in $\mathit{MG}_i$ with the same $id_{seed}$.
This prevents them from staying at different nodes in the next round.
Finally, using $a_i.\mathit{seed}$ and $a_i.\mathit{numRound}$, $a_i$ executes a rendezvous algorithm to meet other good agents.

\subsubsection{Calculate a trust relationship}
\label{subsec:ExecuteUT}
\begin{algorithm}[t]
  \caption{$\mathsf{UpdateTrustRelationship}$}
  \label{alg:UpdateTrustRelationship}
  \begin{algorithmic}[1]
    \State $\mathrm{OVERWRITE}(\mathit{TmpT}_{wit},a_i.T_{wit})$
    \ForAll{$a_j\in \mathit{MG}_i$}
      \State $\mathit{TmpT}_{wit}[a_j.id]\gets \mathsf{UpdateValue}(a_j.id)$
    \EndFor
    \ForAll{$a_j\in \{a_p\mid a_p.id\in a_i.T_{wit}\wedge a_p\notin \mathit{MG}_i\}$}
      \State $\mathrm{REMOVE}(\mathit{TmpT}_{wit},a_j.id)$
    \EndFor
    \State $\mathrm{OVERWRITE}(a_i.T_{wit},TmpT_{wit})$
  \end{algorithmic} 
\end{algorithm}

\begin{algorithm}[t]
  \caption{$\mathsf{UpdateValue}(a_j.id)$}
  \label{alg:UpdateValue}
  \begin{algorithmic}[1]
    \Statex \textbf{Function} $\mathrm{detectAnomaly}(a_j,a_\ell)=(
    a_\ell\notin \mathit{MG}_i\vee 
    a_j.R\neq a_\ell.R\vee 
    a_j.\mathit{numRound}\neq a_\ell.\mathit{numRound}\vee 
    a_j.\mathit{seed}\neq a_\ell.\mathit{seed}\vee 
    a_j.T_{wit}\neq a_\ell.T_{wit}
    )$
    : This function returns whether agent $a_j$ detects an anomaly for agent $a_\ell$.
    \Statex
    \State $\mathrm{OVERWRITE}(\mathit{TmpArray},\emptyset)$
    \Statex //UV-Case (1)
    \ForAll{$a_\ell.id\in a_j.R$}
      \If{$\mathrm{detectAnomaly}(a_j,a_\ell)=\mathit{True}$}
        \State $\mathit{TmpArray}[a_\ell.id]\gets 1$
      \ElsIf{$a_i.T_{wit}[a_j.id][a_\ell.id]<\tau$}
        \State $\mathit{TmpArray}[a_\ell.id]\gets a_j.T_{wit}[a_j.id][a_\ell.id]+1$
      \EndIf
    \EndFor
    \Statex //UV-Case (2)
    \ForAll{$a_\ell\in \{a_p\mid a_p\in \mathit{MG}_i\setminus a_j.R\wedge a_p.id\notin a_j.T_{wit}[a_j.id]\}$}
      \State $\mathit{TmpArray}[a_\ell.id]\gets \tau$
    \EndFor
    \Statex //UV-Case (3)
    \ForAll{$a_\ell\in\{a_p\mid a_p.id\in a_j.T_{wit}[a_j.id]\wedge a_p.id\notin a_j.R\wedge a_j.T_{wit}[a_j.id][a_p.id]<\tau\}$}
      \State $\mathit{TmpArray}[a_\ell.id]\gets a_j.T_{wit}[a_j.id][a_\ell.id]+1$
    \EndFor
    \State \textbf{Return} $\mathit{TmpArray}$
  \end{algorithmic} 
\end{algorithm}

Algorithm \ref{alg:UpdateTrustRelationship} is the pseudo-code of Algorithm $\mathsf{UpdateTrustRelationship}$. 
This algorithm aims to allow an agent $a_i$ to calculate the trust relationship of all agents in $\mathit{MG}_i$.
To do this, $a_i$ uses a two-dimensional array $a_i.T_{wit}$.
The indices in each dimension of $a_i.T_{wit}$ are comprised of agent IDs, and for two agents $a_j$ and $a_\ell$, the value of $a_i.T_{wit}[a_j.id][a_\ell.id]$ indicates the trust level as perceived by $a_i$.
Specifically, if this value is between 1 and $\tau-1$, it implies that $a_i$ thinks $a_j$ does not trust $a_\ell$; otherwise (i.e., if this value is $\tau$), it represents that $a_i$ thinks $a_j$ trusts $a_\ell$.
(Recall that the value is in the range ${1, \dots, \tau}$.)
To update $a_i.T_{wit}$, the algorithm employs a temporary variable $\mathit{TmpT}_{wit}$ to store the updated values.
This ensures that $a_i.T_{wit}$ is not referenced before all elements are updated.
Agent $a_i$ finally replaces $a_i.T_{wit}$ with $\mathit{TmpT}_{wit}$ using $\mathrm{OVERWRITE}()$.

To calculate the trust relationship of all agents in $\mathit{MG}_i$, this algorithm first updates the trust relationship for $a_i$ itself and then for the other agents.
More specifically, $a_i$ first fixes the first-dimension of $a_i.T_{wit}$ to $a_i$ and executes $\mathsf{UpdateValue}(a_i.id)$ to update elements of $a_i.T_{wit}[a_i.id]$.
Algorithm $\mathsf{UpdateValue}(a_j.id)$ returns an array with the updated values for the elements of $a_i.T_{wit}[a_j.id]$.
Following this, $a_i$ updates the elements of $a_i.T_{wit}[a_j.id]$ for an agent $a_j$ in $\mathit{MG}_i\setminus\{a_i\}$ using $\mathsf{UpdateValue}(a_j.id)$.
Finally, $a_i$ removes IDs of all agents not present at the current node from the first dimension of $a_i.T_{wit}$ using $\mathrm{REMOVE}()$.

Next, we explain the details of $\mathsf{UpdateValue}(a_j.id)$ to update all elements in $a_i.T_{wit}[a_j.id]$.
In this algorithm, $a_i$ decides the updated value for an element $a_i.T_{wit}[a_j.id][a_\ell.id]$ for an agent $a_\ell$ as follows.
\begin{description}
  \item[UV-Case (1)] Assume that $a_\ell$ belongs to the same group as $a_j$ at the end of the previous round. If $a_j$ detects an anomaly for $a_\ell$, the updated value is 1. Here, $a_j$ detects an nomaly for $a_\ell$ if any of the following conditions hold: $a_\ell$ does not belong to $\mathit{MG}_i$; or $a_j$ and $a_\ell$ have different values for any of the variables $R$, $\mathit{numRound}$, $\mathit{seed}$, or $T_{wit}$. These discrepancies indicate that $a_\ell$ may not have followed the expected protocol behavior in the previous round. If $a_j$ does not detect an anomaly for $a_\ell$ and $a_j$ does not trust $a_\ell$, the updated value is $a_j.T_{wit}[a_j.id][a_\ell.id]+1$.
  \item[UV-Case (2)] If $a_j$ meets $a_\ell$ for the first time, the updated value is $\tau$.
  \item[UV-Case (3)] If $a_j$ has previously met $a_\ell$, $a_\ell$ does not belong to the same group as $a_j$ in the previous round, and $a_j$ does not trust $a_\ell$, then the updated value is $a_i.T_{wit}[a_j.id][a_\ell.id]+1$.
\end{description}
To verify the condition in UV-Case (1), $a_i$ uses the results of $\mathrm{detectAnomaly}(a_j,a_\ell)$.
In UV-Case (2), where the indices of $a_i.T_{wit}[a_i.id]$ do not include $a_j.id$, $a_i$ implicitly adds $a_j.id$ to the indices of $\mathit{TmpT}_{wit}[a_i.id]$.

In summary, indices in the first-dimension of $a_i.T_{wit}$ only include IDs of agents in $\mathit{MG}_i$.
The updated value of each element $a_i.T_{wit}[a_j.id][a_\ell.id]$ depends on the locations and states of $a_j$ and $a_\ell$ at the beginning of the round.
This gives us the following observation.

\begin{observation}
\label{obs:AgentsAtSameNodeHaveSameTwit}
Every pair of good agents at the same node has the same $T_{wit}$.
\end{observation}

\subsubsection{Select a seed}
\label{subsec:ExecuteSS}
\begin{algorithm}[t]
  \caption{$\mathsf{SelectSeed}$}
  \label{alg:SelectSeed}
  \begin{algorithmic}[1]
    \State $\mathit{TmpR}\gets \{a_j.id\mid$ a node with $a_j.id$ is reachable from a node with $a_i.id$ in $\mathbf{CG}(\mathit{MG}_i,a_i.T_{wit})\}$
    \State $a_i.\mathit{seed}\gets \min(\mathit{TmpR})$ 
    \ForAll{$a_j\in \mathit{TmpR}$} \label{line:SS:CheckNumRoundOfAgentsInConnectedComponet}
      \If{$\mathit{TmpR}\neq a_j.R$}
        \State $a_i.\mathit{numRound}\gets 0$
      \EndIf
    \EndFor \label{line:SS:CheckNumRoundOfAgentsInConnectedComponetEnd}
    \State $a_i.R\gets \mathit{TmpR}$ \label{line:SS:SubstituteConnectedComponentForR}
  \end{algorithmic} 
\end{algorithm}

Algorithm \ref{alg:SelectSeed} is the pseudo-code of Algorithm $\mathsf{SelectSeed}$. 
This algorithm serves two following objectives: (1) an agent $a_i$ selects an ID $id_{seed}$ for use in the current round of the rendezvous algorithm, and (2) agent $a_i$ reaches a consensus on $\mathit{numRound}$ with other agents in $\mathit{MG}_i$ who share the same $id_{seed}$.

To achieve Objective (1), $a_i$ first creates a confidence graph.
Formally, a confidence graph is defined as follows.

\begin{definition}
\label{def:ConfidenceGraph}
Let $\mathbf{CG}(\mathit{MG}_i,T_{wit})$ be the undirected graph whose node set is the set of agents in $\mathit{MG}_i$, and there is an edge between two nodes $id_1$ and $id_2$ if and only if $T_{wit}[id_1][id_2] = T_{wit}[id_2][id_1] = \tau$.  
We call this graph the \emph{confidence graph}.
\end{definition}

\noindent After creating the confidence graph, $a_i$ identifies the IDs that are reachable from the node with $a_i.id$ in $\mathbf{CG}(\mathit{MG}_i,a_i.T_{wit})$ and selects these agents as group members.
Agent $a_i$ stores the IDs of these group members in a temporary variable $\mathit{TmpR}$ to detect any changes in group membership since the end of the previous round.
Finally, $a_i$ selects the smallest ID among IDs of the group members and stores it in $a_i.\mathit{seed}$.

To achieve Objective (2), $a_i$ compares the current group members with the agents in $a_j.R$ for each group member $a_j$.
If there are discrepancies, $a_i$ initializes $a_i.\mathit{numRound}$.
Here, agents in $a_j.R$ means group members of $a_j$ in the previous round.
This comparison checks if each group member belongs to a different group than in the previous round.
Thus, when agents in $a_j.R$ are different from them by incoherence, $a_i$ can detect changes in members of the group with $a_j$.
Therefore, this comparison ensures that all group members have a consistent $\mathit{numRound}$ by the end of a round.

\subsection{Correctness and Complexity}
In this subsection, we prove the correctness and complexity of the algorithm proposed in Section \ref{sec:SS-PG}.
To do this, we first prove that two good agents at the same node make the same confidence graph.
For simplicity of discussion, we define the time at $c_0 \in C_{all}$ as round 1.

\begin{lemma}
\label{lem:IfGoodsExistAtSameNodeTheyMakeSameConfidenceGraph}
For $r \ge 1$, in any round $r$, if two good agents exist at the same node, they create the same confidence graph in $r$.
\end{lemma}
\begin{proof}
Let $a_i$ and $a_j$ be such agents.
Given that both agents are the same node, it follows that $\mathit{MG}_i=\mathit{MG}_j$; therefore, $V_{CG}$ of $a_i$ and $a_j$ are identical.

Without loss of generality, each edge in $E_{CG}$ of $a_i$ is determined by $a_i.T_{wit}$ and $\mathit{MG}_i$.
According to Observation \ref{obs:AgentsAtSameNodeHaveSameTwit}, $a_i.T_{wit}=a_j.T_{wit}$.
Consequently, $E_{CG}$ of $a_i$ and $a_j$ are also identical.
\end{proof}

Next, we prove that if two good agents belong to the same group at the end of a round, they do not detect an anomaly in each other in the next round.

\begin{lemma}
\label{lem:IfGoodBelongToSameGroupTheyDoNotDetectAnAnomalyOfEachOther}
Consider any round $r$ such that $r \ge 2$.
If two different good agents $a_i$ and $a_j$ belong to the same group at the end of $r$, $\mathrm{detectAnomaly}(a_i,a_j)=\mathrm{detectAnomaly}(a_j,a_i)=\mathit{False}$ holds in $r+1$.
\end{lemma}
\begin{proof}
If $a_i$ and $a_j$ have the same $R$ and $\mathit{numRound}$ at the end of $r$, they calculate the same $\mathit{seed}$ by the end of $r$ and exist at a single node in $r+1$.
Additionally, from Observation \ref{obs:AgentsAtSameNodeHaveSameTwit}, $a_i.T_{wit}=a_j.T_{wit}$ holds at the beginning of $r+1$.
Thus, to prove this lemma, it is sufficient to prove that $a_i$ and $a_j$ have the same $R$ and $\mathit{numRound}$ at the end of $r$.

First, we prove that $a_i$ and $a_j$ have the same $R$ at the end of $r$.
From Lemma \ref{lem:IfGoodsExistAtSameNodeTheyMakeSameConfidenceGraph}, $\mathbf{CG}(\mathit{MG}_i,a_i.T_{wit})=\mathbf{CG}(\mathit{MG}_j,a_j.T_{wit})$ holds in $r$.
Without loss of generality, in $\mathbf{CG}(\mathit{MG}_i,a_i.T_{wit})$, the node with label $a_j.id$ is reachable from the node with label $a_i.id$ since $a_i$ and $a_j$ belongs to the same group.
Therefore, $a_i.R=a_j.R$ holds at the end of $r$.

Next, we prove that $a_i$ and $a_j$ have the same $\mathit{numRound}$ at the end of $r$.
From Lines \ref{line:SS:CheckNumRoundOfAgentsInConnectedComponet}-\ref{line:SS:CheckNumRoundOfAgentsInConnectedComponetEnd} of Algorithm \ref{alg:SelectSeed}, in $r$, if $a_i.R=\mathit{TmpR}$, and all agents in $\mathit{TmpR}$ have the same $R$, then all agents in $\mathit{TmpR}$ have the same $\mathit{numRound}$ by Function $\mathrm{detectAnomaly}()$.
Otherwise, all agents in $\mathit{TmpR}$ initialize their $\mathit{numRound}$ by the end of $r$.
Thus, $a_i.\mathit{numRound}=a_j.\mathit{numRound}$ holds at the end of $r$.
\end{proof}

Next, we prove that a good agent does not detect an anomaly in another good agent on or after round 2.

\begin{lemma}
\label{lem:NoGoodExecuteResettingTwitForAnotherGood}
Let $a_i$ and $a_j$ be two good agents.
Agent $a_i$ does not execute $a_i.T_{wit}[a_i.id][a_j.id]\gets 1$ in any round on or after round 2.
\end{lemma}
\begin{proof}
We prove this lemma by contradiction.
Assume that $a_i$ executes $a_i.T_{wit}[a_i.id][a_j.id]\gets 1$ in a certain round after round 2.
Let $r$ be such a round.
It holds that at the beginning of $r$, (1) $a_j\in a_i.R$ and (2) $\mathrm{detectAnomaly}(a_i,a_j)=\mathit{True}$.

From (1), $a_i$ and $a_j$ exist at the same node in $r-1$.
From Observation \ref{obs:AgentsAtSameNodeHaveSameTwit}, $a_i.T_{wit}=a_j.T_{wit}$ holds.
In addition, from (1), the node with label $a_j.id$ is reachable from the node with label $a_i.id$ in $\mathbf{CG}(\mathit{MG}_i,a_i.T_{wit})$ in $r-1$.
This also holds true for $a_j.id$ since $\mathbf{CG}(\mathit{MG}_i,a_i.T_{wit})=\mathbf{CG}(\mathit{MG}_j,a_j.T_{wit})$ holds in $r-1$ from Lemma \ref{lem:IfGoodsExistAtSameNodeTheyMakeSameConfidenceGraph}.
Thus, $a_i$ and $a_j$ belong to the same group at the end of $r-1$.

From Lemma \ref{lem:IfGoodBelongToSameGroupTheyDoNotDetectAnAnomalyOfEachOther}, $\mathrm{detectAnomaly}(a_i,a_j)=\mathit{False}$ holds in $r$, which is inconsistent with (2).
This contradiction proves the lemma.
\end{proof}

From Lemma \ref{lem:NoGoodExecuteResettingTwitForAnotherGood}, we have the following corollary.

\begin{corollary}
\label{cor:GoodTrustAnotherGoodAfterItAddThatToTwit}
Let $a_i$ be a good agent, and $a_j$ be another good agent such that $a_j\in a_i.T_{wit}[a_i.id]$ holds.
Let $r_{ini}^j\geq 1$ be the first round where $a_j\in a_i.T_{wit}[a_i.id]$ holds.
By $r_{ini}^j+\tau$, $a_i.T_{wit}[a_i.id][a_j.id]=\tau$ holds.
\end{corollary}

Next, we prove the number of times a good agent is interrupted while executing a rendezvous algorithm during $\tau$ rounds.

\begin{lemma}
\label{lem:TheNumberOfInterrupt}
Let $a_i$ be a good agent.
For any round $r$ on or after round 2, $a_i$ changes $a_i.\mathit{seed}$ and initializes $a_i.\mathit{numRound}$ at most $3kf$ times between $r$ and $r+\tau$.
\end{lemma}
\begin{proof}
Let $r'$ be a round such that $r\leq r'\leq r+\tau$, and $a_j$ be an agent contained in $a_i.R\cap \mathit{TmpR}$ in $r'$.
By $\mathrm{detectAnomaly}()$, $a_i.R=a_j.R$ holds at the beginning of $r'$.
Thus, if $a_i.R=\mathit{TmpR}$ in $r'$, $a_i$ does not change $a_i.\mathit{seed}$ and does not initialize $a_i.\mathit{numRound}$ in $r'$.
Conversely, $a_i$ changes $a_i.\mathit{seed}$ and initializes $a_i.\mathit{numRound}$ if one of the following cases occur in $r'$: 
(1) $a_i.R\subset \mathit{TmpR}$, i.e., some agents join $a_i$ in $r'$, and
(2) $\mathit{TmpR}\subset a_i.R$, i.e., some agents leave $a_i$ in $r'$.
Hence, to prove this lemma, we show how often Cases (1) and (2) occur between $r$ and $r+\tau$.
To do this, we first consider the number of times that an agent $a_\ell$ causes Cases (1) or (2) between $r$ and $r+\tau$.
We distinguish two cases: $a_\ell$ is a good agent, and $a_\ell$ is a Byzantine agent.

\textbf{Case $a_\ell$ is a good agent}:
When $a_\ell$ joins a member of the group with $a_i$ for the first time, $a_i.R\subset \mathit{TmpR}$ holds, and thus Case (1) occurs by $a_\ell$.
On the other hand, after becoming the member, $\mathrm{detectAnomaly}(a_i,a_\ell)=\mathrm{detectAnomaly}(a_\ell,a_i)=\mathit{False}$ in the next round from Lemma \ref{lem:IfGoodBelongToSameGroupTheyDoNotDetectAnAnomalyOfEachOther}.
Also, from Lemma \ref{lem:NoGoodExecuteResettingTwitForAnotherGood}, $a_i$ and $a_\ell$ do not execute $a_i.T_{wit}[a_i.id][a_\ell.id]\gets 1$ and $a_\ell.T_{wit}[a_\ell.id][a_i.id]\gets 1$ in any round on or after round 2.
Thus, $a_\ell$ does not leave $a_i$; Cases (1) and (2) do not occur by $a_\ell$.
Therefore, in this case, $a_\ell$ causes Case (1) at most once between $r$ and $r+\tau$.

\textbf{Case $a_\ell$ is a Byzantine agent}:
First, we consider that $a_\ell$ causes Cases (1) or (2) by itself.
In this case, $a_\ell$ does so in the following ways:
(a) Agents $a_i$ and $a_\ell$ exist at a single node after belonging to the same group, and 
(b) Agent $a_\ell$ leaves $a_i$ after $a_i$ and $a_\ell$ belong to the same group.
However, for $a_i$ and $a_\ell$ to belong to the same group again, $a_i$ needs $\tau$ rounds after $a_\ell$ leaves.
Thus, $a_\ell$ causes Cases (1) or (2) by itself at most 2 times between $r$ and $r+\tau$.
Next, we consider that $a_\ell$ causes Cases (1) or (2) using another agent $a_m$.
In this case, $a_\ell$ does so in the following ways:
(c) When $a_i$ and $a_\ell$ meet at a single node, $a_i$ and $a_m$ belong to the same group, and $a_\ell$ and $a_m$ trust each other, and 
(d) Agent $a_\ell$ leaves from $a_m$ after $a_i$ and $a_\ell$ belong to the same group by $a_m$ building bridges.
Thus, every time another agent joins $a_i$, $a_\ell$ causes Cases (1) or (2) at most 2 times.
Since the number of agents that could be $a_m$ is at most $k-2$, $a_\ell$ causes Cases (1) or (2) using another agent at most $2(k-2)$ times between $r$ and $r+\tau$.

Finally, we show the number of times that Cases (1) and (2) occur between $r$ and $r+\tau$.
Since the number of good agents other than $a_i$ is $k-f-1$, Cases (1) and (2) occur at most $k-f-1$ times by good agents.
Since the number of Byzantine agents is $f$, Cases (1) and (2) occur at most $2f+2f(k-2)$ times by Byzantine agents.
Hence, the number of times that Cases (1) and (2) occur between $r$ and $r+\tau$ is at most $k-f-1+2f(k-2)<3kf$.
\end{proof}

Next, we prove that two good agents meet within $\tau$ rounds.

\begin{lemma}
\label{lem:TwoGoodMeetDuringTauRounds}
For any round $r$ on or after round 1, any pair of good agents meet between $r$ and $r+\tau$.
\end{lemma}
\begin{proof}
Let $a_i$ and $a_j$ be two good agents.
According to Theorem \ref{the:RendezvousAlgorithm}, if $a_i$ and $a_j$ execute the rendezvous algorithm with different labels for $\TIME{REN}(\min(a_i.\mathit{seed},a_j.\mathit{seed}))$ rounds without changing their $\mathit{seed}$ or initializing their $\mathit{numRound}$, they meet during this period.

Without loss of the generality, Lemma \ref{lem:TheNumberOfInterrupt} implies that $a_i$ changes $a_i.\mathit{seed}$ and initializes $a_i.\mathit{numRound}$ at most $3kf$ times during $\tau$ rounds.
Thus, between $r$ and $r+\tau$, $a_i$ and $a_j$ update these variables at most $6kf$ times in total.

Given that $\tau$ is $(6KF+1)\cdot \TIME{REN}(2^{\Lambda_g+1})$, there are at least $\TIME{REN}(\min(a_i.\mathit{seed},a_j.\mathit{seed}))$ rounds between $r$ and $r+\tau$ during which $a_i$ and $a_j$ do not change their $\mathit{seed}$ or initialize their $\mathit{numRound}$.
Therefore, $a_i$ and $a_j$ meet between $r$ and $r+\tau$.
\end{proof}

Finally, we prove the correctness and complexity of $\mathcal{A}_{sspg}(id, n', k', f', \Lambda_g')$.

\begin{theorem}
Let $N$ be the upper bound on the number of nodes, $K$ be the upper bound on the total number of agents, $F$ be the upper bound on the number of Byzantine agents, and $\Lambda_g$ be the length of the largest ID among IDs of good agents.
If $N$, $K$, $F$, and $\Lambda_g$ are given to agents, Algorithm \ref{alg:SSPGAlgorithm} solves the \texttt{SS-PGAT} with global knowledge in $O(K\cdot F\cdot \TIME{REN}(2^{\Lambda_g}))$ rounds.
\end{theorem}
\begin{proof}
First, we prove that $a_i.T_{wit}[a_i.id][a_j.id]=\tau$ and $a_j.T_{wit}[a_j.id][a_i.id]=\tau$ hold from $\tau+2$ onwards for any pair $a_i,a_j$ of good agents.
From Lemma \ref{lem:TwoGoodMeetDuringTauRounds}, $a_i$ and $a_j$ meet between rounds $2$ and $\tau+1$.
Therefore, $a_i\in a_j.T_{wit}[a_j.id]$ and $a_j\in a_i.T_{wit}[a_i.id]$ hold at the beginning of $\tau+2$.
If $a_i$ (resp. $a_j$) meets $a_j$ (resp. $a_i$) for the first time, $a_i.T_{wit}[a_i.id][a_j.id]=\tau$ holds (resp. $a_j.T_{wit}[a_j.id][a_i.id]=\tau$ holds).
Otherwise, $a_i\in a_j.T_{wit}[a_j.id]$ (resp. $a_j\in a_i.T_{wit}[a_i.id]$) has already held at the beginning of $\tau+2$, and $a_i.T_{wit}[a_i.id][a_j.id]=\tau$ holds (resp. $a_j.T_{wit}[a_j.id][a_i.id]=\tau$ holds) by $\tau+2$ from Corollary \ref{cor:GoodTrustAnotherGoodAfterItAddThatToTwit}.
From Lemma \ref{lem:NoGoodExecuteResettingTwitForAnotherGood}, after $a_i.T_{wit}[a_i.id][a_j.id]=\tau$ holds (resp. $a_j.T_{wit}[a_j.id][a_i.id]=\tau$ holds), $a_i$ (resp. $a_j$) does not execute $a_i.T_{wit}[a_i.id][a_j.id]\gets 1$ (resp. $a_j.T_{wit}[a_j.id][a_i.id]\gets 1$).
Hence, $a_i.T_{wit}[a_i.id][a_j.id]=\tau$ and $a_j.T_{wit}[a_j.id][a_i.id]=\tau$ hold from $\tau+2$ onwards.

From Lemma \ref{lem:TwoGoodMeetDuringTauRounds} and the above discussion, $a_i$ and $a_j$ meet by $2\tau+1$ and belong to the same group in this round.
From Lemma \ref{lem:IfGoodBelongToSameGroupTheyDoNotDetectAnAnomalyOfEachOther}, $a_i$ and $a_j$ stay at a single node from $2\tau+1$ onwards.
Given that $\tau$ is $(6KF+1)\cdot \TIME{REN}(2^{\Lambda_g+1})$, this theorem holds.
\end{proof}

\section{Discussion}
\label{sec:Discussion}
In Section \ref{sec:SS-PG}, we present an algorithm to solve the \texttt{SS-PGAT} without global knowledge.
Here, we discuss (1) removing the assumption of the simultaneous startup, and (2) the memory space required for each agent in this section.

\paragraph*{Removing the assumption of simultaneous startup}
For the execution without simultaneous startup, let the first round, denoted as round $1$, be the round in which the first agent $a$ is selected by the adversary and starts the algorithm.
A similar discussion as Lemma \ref{lem:TwoGoodMeetDuringTauRounds} ensures that within at most $\tau + 1$ rounds, the other good agents meet $a$ and start the algorithm. 
This fact holds since dormant agents stay at the node where they are located, and in the first $\tau$ rounds, there are at least $\TIME{REN}(a.\mathit{seed})$ rounds during which $a$ does not change $a.\mathit{seed}$ and does not initialize $a.\mathit{numRound}$. 
This observation ensures that the good agents start the algorithm within at most $\tau + 1$ rounds. 
After all good agents start the algorithm, the self-stabilizing property of the algorithm guarantees that the good agents successfully meet at a node and subsequently stay at a node within $2\tau + 1$ rounds. 
Therefore, without the assumption of the simultaneous startup, the algorithm solves the \texttt{SS-PGAT} within $3\tau + 1$ rounds, which is asymptotically the same as the algorithm with simultaneous startup. 

\paragraph*{Memory Space}
To analyze the memory space, we consider the maximum size of each variable.
\begin{itemize}
  \item For Variable $\mathit{numRound}$, since it stores at most $\TIME{REN}(2^{\Lambda_g+1})$, its maximum size is given by $\log (\Tilde{O}(N^5\cdot \Lambda_g))$.
  \item For Variable $\mathit{seed}$, since it stores at most ID $2^{\Lambda_g+1}$, its maximum size is given by $\Lambda_g+1$.
  \item For Variable $T_{wit}$, the indices of its first dimension accommodate at most $k$ IDs, while those of the second dimension store at most $2^{\Lambda_g+1}$ IDs. 
  Every element stores at most $(6KF+1)\cdot \TIME{REN}(2^{\Lambda_g+1})$. 
  Thus, its maximum size is $O(k\cdot 2^{\Lambda_g}\cdot\log (K\cdot F\cdot \TIME{REN}(2^{\Lambda_g}))$.
  \item For Variable $R$, it stores at most $k$ IDs, thus its maximum size is $O(k\cdot \Lambda_g)$.
\end{itemize}
In summation, the required memory space is $O(k\cdot 2^{\Lambda_g}\cdot\log (K\cdot F\cdot \TIME{REN}(2^{\Lambda_g}))$, and this result is primarily determined the maximum size of Variable $T_{wit}$.

The memory space can be reduced by additional operations.
For each index $i$ of the first dimension of $T_{wit}$, an agent ensures that indices of $T_{wit}[i]$ comprise at most $K$ IDs.
More concretely, when an agent stores a new ID in the indices of $T_{wit}[i]$, it deletes the ID of the least seen agent from indices of $T_{wit}[i]$ if the number of indices of $T_{wit}[i]$ exceeds $K$.
This operation allows agents to eliminate the IDs of non-existent agents introduced by incoherence, and thus the maximum size of $T_{wit}$ becomes $O(k\cdot K\cdot\log (K\cdot F\cdot \TIME{REN}(2^{\Lambda_g}))$, that is, becomes smaller.
Thus, memory space can be lowered to $O(k\cdot K\cdot\log (K\cdot F\cdot \TIME{REN}(2^{\Lambda_g}))$.

\section{Conclusion}
In this paper, we presented two impossibility results regarding the solvability of the \texttt{GAT} and \texttt{SS-GAT}, and proposed an algorithm to solve the \texttt{SS-PGAT}.
We thoroughly analyzed the global knowledge required to solve the \texttt{GAT} and \texttt{SS-GAT}, examining each case with and without such knowledge.
Based on the impossibility results, we introduced the \texttt{PGAT} as a relaxed variant of the \texttt{GAT} and investigated the global knowledge required for this problem as well.

Our findings show that no algorithm can solve the \texttt{GAT} without global knowledge, and no algorithm can solve the \texttt{SS-GAT} for any $k$ and $f$ with $n$, $k$, $f$, and $\Lambda_g$. 
On the other hand, for the \texttt{SS-PGAT}, we provided an algorithm that solves the problem with $N$, $K$, $F$, and $\Lambda_g$, where these values represent the upper bounds of $n$, $k$, and $f$, respectively.
The time complexity of our proposed algorithm is $O(K\cdot F\cdot \TIME{REN}(2^{\Lambda_g}))$.

Despite these results, the solvability of the \texttt{SS-PGAT} without global knowledge remains an open problem.

\section*{Acknowledgement}
This work was supported by JSPS KAKENHI Grant Numbers JP20KK0232, JP22H03569, JP25K03078, JP25K03079, and JP25K03101
and JST FOREST Program JPMJFR226U. 
\bibliographystyle{unsrt}

\end{document}